\documentclass[11pt]{article}
\usepackage[margin=1in]{geometry}

\usepackage{amsmath,amsthm,amsfonts,amssymb}
\usepackage{graphicx}
\usepackage{url}
\usepackage{lineno}
\usepackage{hyperref}

\usepackage{xcolor}

\title{Upper bounds for stabbing simplices by a line}
\date{}
\author{Inbar Daum-Sadon\thanks{\texttt{inbar.sadon@gmail.com}. Ariel University, Ariel, Israel.} \and Gabriel Nivasch\thanks{ \texttt{gabrieln@ariel.ac.il}. Ariel University, Ariel, Israel.}}

\newcommand{\R}{\mathbb R}
\newcommand{\stconv}{\mathrm{stconv}}
\newcommand{\conv}{\mathrm{conv}}
\newcommand{\BB}{\mathrm{BB}}
\newcommand{\RecFSG}{\mathrm{RecFSG}}
\newcommand{\RecFSD}{\mathrm{RecFSD}}
\newcommand{\FSD}{\mathrm{FSD}}

\newtheorem{theorem}{Theorem}
\newtheorem{lemma}[theorem]{Lemma}
\newtheorem{observation}[theorem]{Observation}

\begin{document}
\maketitle

\begin{abstract}
It is known that for every dimension $d\ge 2$ and every $k<d$ there exists a constant $c_{d,k}>0$ such that for every $n$-point set $X\subset \R^d$ there exists a $k$-flat that intersects at least $c_{d,k} n^{d+1-k} - o(n^{d+1-k})$ of the $(d-k)$-dimensional simplices spanned by $X$. However, the optimal values of the constants $c_{d,k}$ are mostly unknown. The case $k=0$ (stabbing by a point) has received a great deal of attention.

In this paper we focus on the case $k=1$ (stabbing by a line). Specifically, we try to determine the upper bounds yielded by two point sets, known as the \emph{stretched grid} and the \emph{stretched diagonal}. Even though the calculations are independent of $n$, they are still very complicated, so we resort to analytical and numerical software methods. We provide strong evidence that, surprisingly, for $d=4,5,6$ the stretched grid yields better bounds than the stretched diagonal (unlike for all cases $k=0$ and for the case $(d,k)=(3,1)$, in which both point sets yield the same bound). Our experiments indicate that the stretched grid yields $c_{4,1}\leq 0.00457936$, $c_{5,1}\leq 0.000405335$, and $c_{6,1}\leq 0.0000291323$.
\end{abstract}

\section{Introduction} 

A $k$-dimensional simplex is the convex hull of $k+1$ affinely independent points in $\R^d$, $d\ge k$. The $k+1$ points are said to \emph{span} the simplex. The following result was proven for the planar case by Boros and F\"{u}redi \cite{Boros1984} and for arbitrary dimension by B\'{a}r\'{a}ny \cite{BARANY1982141}: \emph{For every $d\geq 2$ there exists a constant $c_d>0$ such that for every $n$, if $X$ is any $n$-point set in $\R^d$ in general position, then there exists a point $x$ in $\R^d$ contained in at least $c_dn^{d+1}-o(n^{d+1})$ full-dimensional simplices spanned by $X$, where $c_d>0$ is a constant depending only on $d$.} Matou\v{s}ek \cite{matousek2002lectures} called this result the \emph{First Selection Lemma}.
It can be used to construct so-called \emph{weak $\varepsilon$-nets} (see \cite{matousek2002lectures}).

The problem of determining largest possible values of the constants $c_d$ has sparked a lot of interest. For the planar case, Boros and F\"{u}redi \cite{Boros1984} showed that $c_2\geq 1/27$. For arbitrary dimension, B\'{a}r\'{a}ny \cite{BARANY1982141} proved  that $c_d\geq\ 1/((d+1)!(d+1)^d)$. Wagner \cite{wagner2003k} subsequently improved this lower bound to $c_d \geq (d^2+1)/((d+1)!(d+1)^{d+1})$. In particular, $c_3\geq 0.001627$. Basit et al.~\cite{Basit:2010:IFS:1810959.1811017} then improved the bound for $c_3$ to $c_3 \geq 0.00227$.

Later, Gromov improved the general lower bound to $c_d\geq 2d/((d+1)!^2(d+1))$ \cite{Gromov2010} (see simpler expositions of this result by Karasev \cite{karasevsimpler} and Jiang \cite{jiang}). This is an improvement by roughly a factor of $e^d$ over the previous bound. In particular, $c_3\geq 0.002604$. Matou\v{s}ek and Wagner \cite{Matousek:1330159} then showed that $c_3 \geq 0.00263$. Later, Kr\'{a}l et al. \cite{Kral2012} slightly improved Gromov's bound for general $d$, yielding in particular $c_3 \geq (3-\sqrt{2})/512 \simeq 0.00309$.

Regarding upper bounds, K\'arteszi~\cite{karteszi} (for $d=2$) and B\'ar\'any~\cite{BARANY1982141} (for general $d$) proved that if $X$ is \emph{any} point set in general position in $\R^d$, then no point in $\R^d$ is contained in more than $n^{d+1}/(2^d (d+1)!) + O(n^d)$ simplices spanned by $X$. Hence, $c_d \leq 1/(2^d (d+1)!)$. This upper bound is ``trivial'' in the sense that it does not rely any specific construction for $X$.

Bukh et al.~obtained the first ``non-trivial'' upper bounds, by constructing, for every $n$ and $d$, a specific point set $X\subset \R^d$ that witnesses $c_d\leq (d+1)^{-(d+1)}$ \cite{Bukh2010}. These are the best upper bounds currently known. The set $X$ is the so-called \emph{stretched diagonal} (presented below). Another point set, called the \emph{stretched grid}~\cite{Bukh2011} (also presented below) gives the same upper bound.

Hence, $c_2=1/27$ is tight, and $c_3\leq 0.0039$. Thus, for $d\geq 3$ the optimal value of $c_d$ is not known, and there is a gap of a factor of roughly $d^d$ between the best lower and upper bounds.

(Some authors prefer to talk about the constant $c'_d$ such that there exists a point in at least $c'_d\binom{n}{d+1}-O(n^d)$ simplices. Then the relation between the two constants is that $c'_d=c_d\cdot (d+1)!$.)

\subsection{Generalization of the First Selection Lemma}

The First Selection Lemma can be generalized as follows. \emph{If $X\subseteq \R^d$ is an $n$-point set in general position, and $k$ is an integer, $0 \leq k < d$, then there exists a $k$-flat that intersects at least $c_{d,k}n^{d-k+1}-O(n^{d-k})$ of the $(d-k)$-dimensional simplices spanned by $X$, for some positive constants $c_{d,k}$ that depend only on $d$ and $k$.} A trivial projection argument yields $c_{d,k}\ge c_{d-k}$.
The problem of determining the maximum values of the constants $c_{d,k}$ was raised by Bukh et al.~\cite{Bukh2010}.

\subsection{Lower and upper bounds for $c_{d,k}$}
The case $k=d-1$ is trivial: An optimal hyperplane is one that partitions the given point set into two equal parts. Hence, $c_{d,d-1}=1/4$.

By a simple projection argument, the above-mentioned result of K\'artesi and B\'ar\'any yields the ``trivial'' upper bound of $c_{d,k} \le 1/(2^{d-k} (d-k+1)!)$.

For the case $k=d-2$, it was shown in \cite{Bukh2010} that there exists a $(d-2)$-flat
that stabs at least $c_{d,d-2}n^3-O(n^2)$ of the triangles spanned by $X$, with
\begin{equation*}
   c_{d,d-2}\geq \frac{1}{24} \Big(1- \frac{1}{(2d-1)^2}\Big).
\end{equation*}
In particular, for $d=3$ there always exists a line that stabs at least $n^3/25 - O(n^2)$ triangles. 

For the case $(d,k)=(3,1)$, Bukh claimed without providing details that in the \emph{stretched grid} every line stabs at most $n^3/25+o(n^3)$ triangles, and therefore $c_{3,1}=1/25$ is tight (this is mentioned in \cite{Gabriel_thesis}).

\subsection{Related problems}

Ashok, Rajgopal and Govindarajan \cite{DBLP:journals/corr/AshokRG14} studied variants of the First Selection Lemma for other classes of geometric objects, such as spheres and axis-parallel boxes in $\R^d$, and  quadrants and slabs in the plane. They also considered
a \emph{strong} variant of the First Selection Lemma, where the
piercing point must come from the point set itself. 

The \emph{Second Selection Lemma} is a generalization of the First Selection Lemma. It states that for every $n$, if $X$ is an $n$-point set in $\R^d$ and $F$ is a family of $\alpha\binom{n}{d+1}$ $X$-simplices, then there exists a point contained in at least $b_d\alpha^{s_d}\binom{n}{d+1}$ simplices of $F$, for some constants $b_d>0$ and $s_d$.

The Second Selection Lemma was conjectured, and proved in the planar case, by B\'ar\'any, F\"uredi and Lov\'asz \cite{Barany1990} (see also Matou\v{s}ek \cite{matousek2002lectures}). A proof for the planar case by a different technique, with considerably better quantitative bounds, was given by Aronov et al.~\cite{Aronov1991}. This bound was then slightly improved by Eppstein, Nivasch, and Sharir \cite{eppstein1993improved,NIVASCH2009494}.
The full proof of the Second Selection Lemma for arbitrary dimension was put together by B\'ar\'any et al.~\cite{Barany1990}, Alon et al.~\cite{alon_barany_furedi_kleitman_1992}, and {\v{Z}}ivaljevi{\'c}
and Vre{\'c}ica \cite{vzivaljevic1992colored}.

The Second Selection Lemma has been used to bound the number of \emph{$k$-sets} in arbitrary dimension, where a $k$-set of a point set $X$ is a subset of $X$ of size $k$ that can be separated from the rest of $X$ by a hyperplane.

Several variants of the Second Selection Lemma, involving geometric objects other than simplices, were proved by Chazelle et al.~\cite{chazelle1994selecting}, Sharir and Smorodinsky  \cite{smorodinsky_sharir_2004}, and Ashok et al.~\cite{DBLP:journals/corr/AshokRG14}.

A similar problem, of \emph{centerline depth}, has been studied by Magazinov and P\'or \cite{MP} and Blagojevi\'c, Karasev, and Magazinov \cite{BKM}.

\subsection{Our results}
In this work we try to determine the upper bounds for the constants $c_{d,1}$ given by the stretched grid and the stretched diagonal. 

For $d=3$, we find that both point sets yield $c_{3,1}\leq 1/25$ according to analytical software methods (as Bukh had already claimed for the stretched grid). Surprisingly, however, for $4\leq d\leq 6$ we find very strong numerical evidence that the stretched grid yields a better bound than the stretched diagonal: On the one hand, for the stretched diagonal there always exists a line that stabs at least $n^d/(d+2)^{d-1}- o(n^d)$ simplices.
On the other hand, the stretched grid likely yields $c_{4,1}\leq 0.00457936$, $c_{5,1}\leq 0.000405335$, and $c_{6,1}\leq 0.0000291323$, according to non-rigorous numerical optimization methods.

\paragraph{Organization of this paper.}
Section~\ref{sec_stconv} reviews the stretched grid and the stretched diagonal, as well as \emph{stair-convexity}, the framework used to analyze them. Section~\ref{sec_grid} presents our results regarding the stretched grid. Section~\ref{sec_diag} presents our results regarding the stretched diagonal. We conclude with some remarks in Section~\ref{sec_final}.

\section{Stair-convexity}\label{sec_stconv}

Following Bukh et al.~\cite{Bukh2011} we define the \emph{stretched grid} as an axis-parallel grid of points in $\R^d$ where, in each axis direction $i$, $2\leq i\leq d$, the spacing between consecutive ``layers'' increases rapidly, and furthermore, the rate of increase for direction $i$ is much larger than that for direction $i-1$. To simplify  calculations, we make the coordinates increase rapidly also in the first direction. We denote the stretched grid by $G_s$. Hence, $G_s(n^d)\subseteq \R^d$ is of the form $G_s(n^d)=X_1\times\ldots\times X_d$ where each $X_i\subseteq \R$ is of the form $X_i=\{x_{i,1},\ldots,x_{i,n}\}$, where $x_{i,1}<x_{i,2}<\cdots<x_{i,n}$ is the $i$th axis that contains $n$ points. We define the sets $X_i$ by induction on $i$, together with relations $\ll_i$ on $\R$,
which describe ``at least how fast'' the terms in $X_i$ must grow (but we will also
use $\ll_i$ for comparing real numbers other than the members of $X_i$). 
We start by letting $x\ll_1y$ mean $K_1x\leq y$, where $K_1=2^d$. Then we choose $X_1$ so that $x_{1,1}=1$ and $x_{1,1}\ll_1 x_{1,2}\ll_1 \cdots \ll_1 x_{1,n}$.
Having defined $X_{i-1}$ and $\ll_{i-1}$ , we set $K_i=2^d x_{(i-1),n}$, we define $x\ll_i y$ to mean $K_ix\leq y$, and we choose $X_i$ so that $x_{i,1}=1$ and $x_{i,1}\ll_i x_{i,2}\ll_i \cdots \ll_i x_{i,n}$.

The \emph{stretched diagonal} is the following subset of the stretched grid:
\begin{equation*}
D_s(n)=\{(x_{1,j},\ldots,x_{d,j})\in \R^d:j=1,2,\ldots,n\}.
\end{equation*}
In other words, the $j$th point of the stretched diagonal is built from the $j$th element of each of $X_1, \ldots, X_d$.

Define the \emph{uniform grid} in the unit cube $[0,1]^d$ by
\begin{equation*}
 G_u=G_u(n^d)=\left\{0,\frac{1}{n-1},\frac{2}{n-1},\ldots,\frac{n-1}{n-1}\right\}^d.
\end{equation*}
Let $\BB(G_s)= [1,x_{1,n}]\times[1,x_{2,n}]\times\cdots\times[1,x_{d,n}]$ be the bounding box of $G_s$,
and let $\pi : \BB(G_s)\rightarrow [0,1]^d$
be a bijection that maps $G_s$ onto $G_u$ and preserves
ordering in each coordinate (that is, we map points of $G_s$ to the corresponding
points of $G_u$ and we squeeze the “elementary boxes” of $G_s$ onto the corresponding elementary boxes of $G_u$). See Figure \ref{fig:Stretched grid}.
\begin{figure}
\centering
\includegraphics{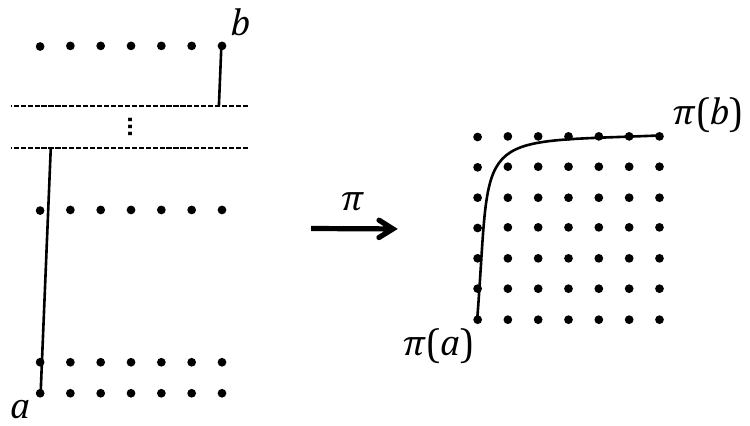}
\caption{\label{fig:Stretched grid} The stretched grid and the mapping $\pi$.
}
\end{figure}
Let us consider the effect of $\pi$ on a straight-line segment $u=ab$ connecting two grid points $a,b\in G_s$. Suppose without loss of generality that $b_d\ge a_d$. Since $G_s$ is so much more stretched in the $d$th direction than in all the previous directions, $\pi(u)$ ascends in the $d$th direction from $\pi(a)$, reaching almost the height of $\pi(b)$, before moving significantly in any other direction. From there on, we can continue tracing $\pi(u)$ by induction on $d$. This observation motivates the notion of \emph{stair-convexity}.

\subsection{Stair-convexity}
Given a pair of points $a,b\in \R^d$, define the \emph{stair-path}  $\sigma(a,b)$ between them as a polygonal path connecting $a$ and $b$ and consisting of at most $d$ closed line segments, each parallel to one of the coordinate axes.
The definition goes by induction on $d$; for $d=1$, the stair-path $\sigma(a,b)$ is simply the segment $ab$. For $d\geq2$, after possibly interchanging $a$ and $b$, let us assume $a_d\leq b_d$. We set $a'=(a_1,\ldots,a_d-1,b_d)$, and we let $\sigma(a,b)$ be the union of the segment $aa'$ and the stair-path $\sigma(a',b)$; for the latter we use induction, ignoring the common last coordinate of $a'$ and $b$.

We call a set $S\subseteq \R^d$ \emph{stair-convex} if for every $a,b \in S$ we have $\sigma(a,b) \subseteq S$.
We define the \emph{stair-convex hull} of a set $S\subseteq \R^d$ as the intersection of all stair-convex sets containing $S$, and we will denote it as $\stconv (S)$. 

\subsection{Intersection of stair-convex hulls of two sets}\label{Intersection}
In stair-convexity we will call the last coordinate of a point its ``height''. For a real number $y$, let $h(y)$ denote the \emph{horizontal hyperplane} $\{x \in \R^d:x_d=y$\}. For a horizontal hyperplane $h = h(y)$, let $h^+=\{x \in \R^d: x_d\geq y\}$ be the upper closed half-space bounded by $h$, and similarly let $h^-$ be the lower closed half-space. For a set $S\subseteq \R^d$, let $S(y)=S\cap h(y)$ be the horizontal slice of $S$ at height $y$. For a point $x= (x_1,\ldots, x_d)\in \R^d$, let $\overline{x}= (x_1,\ldots, x_{d-1})$ be the projection of $x$ into $\R^{d-1}$, and define $\overline{S}$ for $S\subset \R^d$ similarly. For a point $x\in \R^{d-1}$ and a real number $x_d$, let $x\times x_d=(x_1,\ldots,x_{d-1}, x_d)$.

\begin{lemma}[\cite{Bukh2011}]\label{Lemma 5.1}
A set $S\subseteq \R^d$ is stair-convex if and only if the following two conditions hold:
\begin{enumerate}
    \item Every horizontal slice $\overline{S(y)}$ is stair-convex.
    \item For every $y_1\leq y_2 \leq y_3$ such that $\overline{S(y_3)} \neq \emptyset$
we have $\overline{S(y_1)}\subseteq \overline{S(y_2)}$. (Meaning, the horizontal slice can only grow with increasing height, except that it can end by disappearing abruptly).
\end{enumerate}

\end{lemma}
\begin{lemma}[\cite{Bukh2011}]\label{Lemma 5.2}
The stair-convex hull of a set $X\subseteq \R^d$ can be (recursively) characterized as follows. For every horizontal hyperplane $h=h(y)$ that does not lie entirely above $X$, let $X'$ stand for the vertical projection of $X\cap h^-$ into $h$. Then $h\cap \stconv(X)= \stconv(X')$ (where $\stconv(X')$ is a stair-convex hull in dimension $d-1$).
\end{lemma}

The following lemma specifies under which conditions the stair-convex hulls of two sets intersect. Recall that in standard geometry, Kirchberger's theorem \cite{kirchberger1903tchebychefsche} states that if $Y$ and $Z$ are point sets in $\R^d$ such that $\conv(Y)$ and $\conv(Z)$ intersect, then there exist
subsets $Y'\subseteq Y$ and $Z' \subseteq Z$ of total size $|Y|+|Z|\leq d+2 $ such that $\conv(Y)$ and $\conv(Z)$ intersect.

\begin{lemma}[\cite{Bukh2011}]\label{Lemma 5.4}
Let $Y, Z \subset \R^d$ be two finite point sets that do not share any coordinate, with $|Y|= s$ and $|Z|=t$. Then:
\begin{enumerate}
\item If $s + t < d + 2$, then $\stconv(Y)$ and $\stconv(Z)$ do not intersect.
\item If $s + t = d + 2$ and $\stconv(Y), \stconv(Z)$ intersect, then they do so at a single point. Suppose they do intersect. Then the two highest points of $Y\cup Z$ (in last coordinate) belong one to $Y$ and one to $Z$. Furthermore, let $y_{\mathrm{top}}, z_{\mathrm{top}}$ be the highest points of $Y,Z$ respectively, and say $y_{\mathrm{top},d}>z_{\mathrm{top},d}$. Then the point of intersection between $\stconv(Y), \stconv(Z)$ is $p = q\times z_{\mathrm{top},d}$, where $q\in\R^{d-1}$ is the point of intersection of $\stconv{\left(\overline{Y\setminus \{y_{\mathrm{top}}\}}\right)}$ and $\stconv(\overline Z)$.
\item If $s+t>d+2$ and $\stconv(Z), \stconv(Y)$ intersect, then there exist subsets $Z'\subseteq Z,Y' \subseteq Y$ of total size $|Z'|+|Y'|=d+2$, such that $\stconv(Z'), \stconv(Y')$ intersect.
\end{enumerate}
\end{lemma}

The special case $|Z|=1$ of Lemma~\ref{Lemma 5.4} is important enough to be stated separately. Let $a,b \in \R^d$ be two points that do not share any coordinate. We say that $b$ has \emph{type $0$} with respect to $a$ if $b_i< a_i$ for every $i = 1,2,\ldots,d$. For $j\in\{1,2,...,d\}$ we say that $b$ has \emph{type $j$} with respect to $a$ if $b_j> a_j$ but $b_i< a_i$ for all $i=j+1,...,d$.
\begin{lemma}[\cite{Bukh2011}]\label{Lemma x in stconv}
Let $X\subseteq \R^d$ be a point set, and let $a \in \R^d$ be a point. Then $a \in \stconv(X)$ if and only if $X$ contains a point of type $j$ with respect to $a$ for every $j=0,1,\ldots,d$.
\end{lemma}

\paragraph{Transference Lemma.}
The almost-correspondence between convex hulls and stair-convex hulls in the stretched grid is formalized in the following lemma. Let us say that two points $a = (a_1,\ldots,a_d)$ and $b = (b_1,\ldots,b_d)$ in $\BB(G_s)$ are \emph{far apart} if, for every $i = 1, 2,...,d$, we have either $a_i\ll_i b_i$ or $b_i\ll_i a_i$. We also extend this notion to sets: Two sets $Y,Z \subseteq \R^d$ are \emph{far apart} if each $z \in Z$ is far apart from each $y \in Y$.

\begin{lemma}[\cite{Bukh2011}]\label{lemma 1.4}
Let $Y,Z$ be sets in $\BB(G_s)$ that are far apart. Then $\stconv(Y) \cap \stconv(Z) = \emptyset$ if and only if $\conv(Y) \cap \conv(Z) = \emptyset$.
\end{lemma}

If a small set $Y\subset\BB(G_s)$ is given, not necessarily from the stretched grid, and we consider all possible fixed-size sets $Z\subset G_s$, then almost all such sets $Z$ will be far apart from $Y$, except for a negligible fraction of them. What will interest us is whether $\conv(Z)$ and $\conv(Y)$ intersect. So, according to the lemma above, in the vast majority of cases it is enough to check whether $\stconv(Z)$ and $\stconv(Y)$ intersect.

\subsection{Warm-up: Upper bounds for the First Selection Lemma}
As a warm-up, we recall the proof that both the stretched grid and the stretched diagonal yield the upper bound $c_d \le (d+1)^{-(d+1)}$ for the First Selection Lemma.

Let $X$ be either the stretched grid or the stretched diagonal, and let $a\in\R^d$ be a point. For each $0\le j\le d$, let $C_j(a)$ be the set of all points $b\in X$ that have type $j$ with respect to $a$. By the Transference Lemma and Lemma \ref{Lemma x in stconv}, the number of full-dimensional simplices spanned by $X$ that contain $a$ is very close to the product $\prod_{j=0}^d |C_j(a)|$. By the arithmetic-geometric mean inequality, this expression achieves its maximum when all terms have the same size, namely $|C_j(a)|=|X|/(d+1)$ for each $j$. The claim follows.

\begin{figure}
\centering
\includegraphics{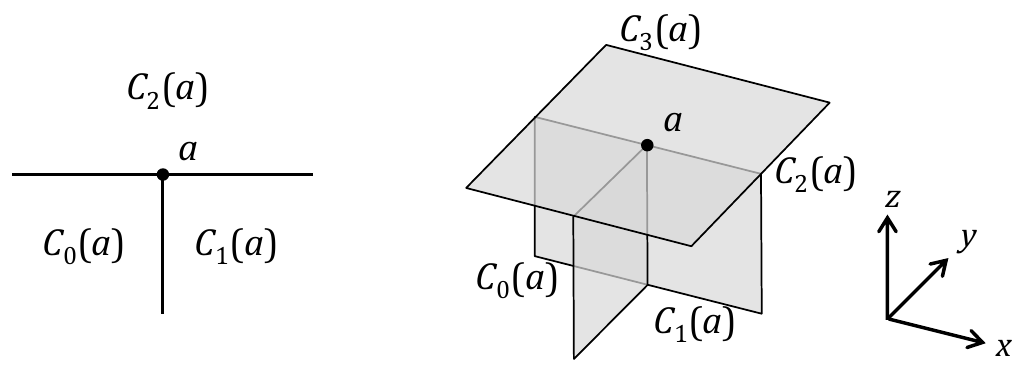}
\caption{\label{fig:Stair_convexity} Partition of space into $d+1$ parts.}
\end{figure}

\section{Results for the stretched grid}\label{sec_grid}

In this section we derive the upper bounds for $c_{d,1}$ yielded by the stretched grid.

\subsection{A recursive formula}

Let $q,p\in [0,1]^d$ be two given points with $q=(q_1, \ldots, q_d)$, $p=(p_1, \ldots, p_d)$. Informally, we want to define the probability $\RecFSG_d(q,p)$ that a randomly chosen $(d-1)$-dimensional stair-simplex from $[0,1]^d$ intersects the stair-path $\sigma(q,p)$. Formally, let $\mathcal{D}$ be the uniform distribution in $([0,1]^d)^d$. Every element $A \in ([0,1]^d)^d$ represents a $d$-tuple $z_1, \ldots, z_d$ of points in $[0,1]^d$ which span the stair-simplex $S(A)=\stconv\{z_1, \ldots, z_d\}$. Then define $\RecFSG_d(q,p)$ as the measure
\begin{equation*}
\RecFSG_d(q,p)=\mu[A\in \mathcal{D}: S(A)\cap \sigma(q,p)\neq\emptyset].
\end{equation*}
The connection between $\RecFSG_d$ and the stretched grid is as follows. Given $q,p$ as above, let $\alpha = \RecFSG_d(q,p)$. Given $n$, let $G_s(n)= X_1 \times \cdots \times X_d$ be the $d$-dimensional stretched grid of dimensions $m \times \cdots \times m$ with $m=n^{1/d}$, where $X_i=  \{x_{i,1}, \ldots, x_{i,m}\}$ for each $i$. If $q'=(q'_1, \ldots, q'_d)$, $p'=(p'_1, \ldots, p'_d)$ are points satisfying $x_{i,q_i m}\le q'_i\le x_{i,q_i m+1}$ and $x_{i,p_i m}\le p'_i\le x_{i,p_i m+1}$, then the stair-path $\sigma(q',p')$ intersects an $(\alpha\pm o(1))$-fraction of the $(d-1)$-dimensional stair-simplices spanned by $G_s(n)$. Let $Y=\{q,p\}$, and let $Z$ be the set of vertices of a stair-simplex. The fraction of stair-simplices for which $Y,Z$ are not far apart is negligible as $n \xrightarrow{}\infty$. Hence, by the Transference Lemma (Lemma \ref{lemma 1.4}), the segment $qp$ also intersects an $(\alpha\pm o(1))$-fraction of the $(d-1)$-dimensional simplices spanned by $G_s(n)$. Since the number of simplices spanned by $G_s(n)$ is $\binom{n}{d} = n^d/d!-O(n^{d-1})$, the constant that multiplies $n^d$ is $\RecFSG_d(q,p)/d!$.

\paragraph{Types of stair-paths.} We define the type $T$ of a stair-path $qp$ by $T=\{j:q_j<p_j\}$ (we can safely ignore cases where $p_j=q_j$ for some $j$, since they have measure $0$).
In dimension $d$ there are $2^{d}$ possible types of stair-paths, but half of them are equivalent to the other half since $p$ and $q$ just switch positions. This leaves us with $2^{d-1}$ possible types of stair-paths. 

 In addition, stair-convexity is symmetric with respect to the first coordinate, and so is the stretched grid. Therefore, without loss of generality we can assume that $1\notin T$. This leaves us with $2^{d-2}$ possible types of stair-paths. 
\begin{theorem}\label{theorem1}
Let $q,p \in [0,1]^d$ be two points, with $q=(q_1, \ldots ,q_d),p=(p_1, \ldots ,p_d)$. If $p_d\geq q_d$ let $x=p,y=q$; otherwise, let $x=q,y=p$. Then $\RecFSG_d$ is given by:
\begin{align*}
\RecFSG_1(q,p)&=x_1-y_1, \\
\RecFSG_d(q,p)&=
d!(x_d^d-y_d^d)\prod_{i=1}^{d-1}{y_i^i(1-y_i)}\\
&\qquad{}+d(1-x_d)x_d^{d-1}\RecFSG_{d-1}(\{\overline{q},\overline{p}\}), \qquad \text{for } d\geq 2.
\end{align*}
\end{theorem}
\begin{proof}
Let $Y=\{q,p\}$, and let $Z=\{z_1, z_2,\ldots ,z_d\}\subset [0,1]^d$ such that $Y,Z$ do not share any coordinate, where $|Y\cup Z|=d+2$. Hence, $\stconv (Y)$ is a stair-path and $\stconv (Z)$ is a $(d-1)$-stair-simplex.
By Lemma \ref{Lemma 5.4} part (2), $\stconv (Y)$ and $\stconv (Z)$ will intersect in at most one point, and if they do intersect, then after projecting the $d+1$ lower points to dimension $d-1$, there will also be an intersection point.

\textbf{Base case:}
When $d=1$, the path is of type $T=\emptyset$ which means $q_1 \geq p_1$. The measure of simplices in $\mathcal{D}$ whose single point lies between them is $q_1-p_1$. 

\textbf{Recursive case:}
 The recursive function is built out of a two-part addition: a non-recursive part that we get when the highest point belongs to $Y$, and a recursive part that we get when the highest point belongs to $Z$.
The first part is derived as follows. Let $p=y_{\mathrm{top}}$ be the highest of all points; see Figure \ref{fig_r3}(\emph{a}).
All the points of $Z$ must be below $p_d$ but not all of them should be below 
$q_d$. The measure of simplices in $\mathcal{D}$ with this property is $p_d^d-q_d^d$.
Now, when projecting to the lower dimension, $d-1$, we ``discard" the highest point $p$, and stay with $\overline{Z}$ and the point $\overline{q}$. So, it remains to calculate the measure of simplices $\stconv (\overline{Z})$ that intersect the point $\overline q$.
Since $z_{\mathrm{top}}\in \overline{Z}$ should be above $q_{d-1}$ and the other $d-1$ points should be below it, this occurs with measure $d(1-q_{d-1})q_{d-1}^{d-1}$.
Let us again project to a lower dimension and ``discard" the highest point. We are left with a simplex in a lower dimension, $\stconv \left(\overline{\overline Z\setminus \{z_{\mathrm{top}}\}}\right)$, and the point $\overline{\overline{q}}$. We continue this way until we reach $d=1$.
Hence, for the first part we get the term
\begin{equation*}
(p_d^d-q_d^d)d!\prod_{i=1}^{d-1}{q_i^i(1-q_i)}.
\end{equation*}
The second part is derived as follows:  Let $q_d\leq p_d\leq z_{\mathrm{top},d}$. The point
$z_{\mathrm{top}}\in Z$ must be above $p_d$, while the other $d-1$ points of $Z$ must be below it; see Figure \ref{fig_r3}(\emph{b}, \emph{c}). This happens with measure $d(1-p_d)p_d^{d-1}$. When we ``discard" the highest point, we are left with a stair-path and a stair-simplex one dimension lower. Therefore, we can recursively invoke $\RecFSG_{d-1}$. Hence, for the second part we get the term
\begin{equation*}
d(1-p_d)p_d^{d-1}\RecFSG_{d-1}(\overline{q},\overline{p}).
\end{equation*}
If $q_d\geq p_d$, all calculations are the same, except that we interchange $q$ and $p$.
\end{proof}

\begin{figure}
\centerline{\includegraphics{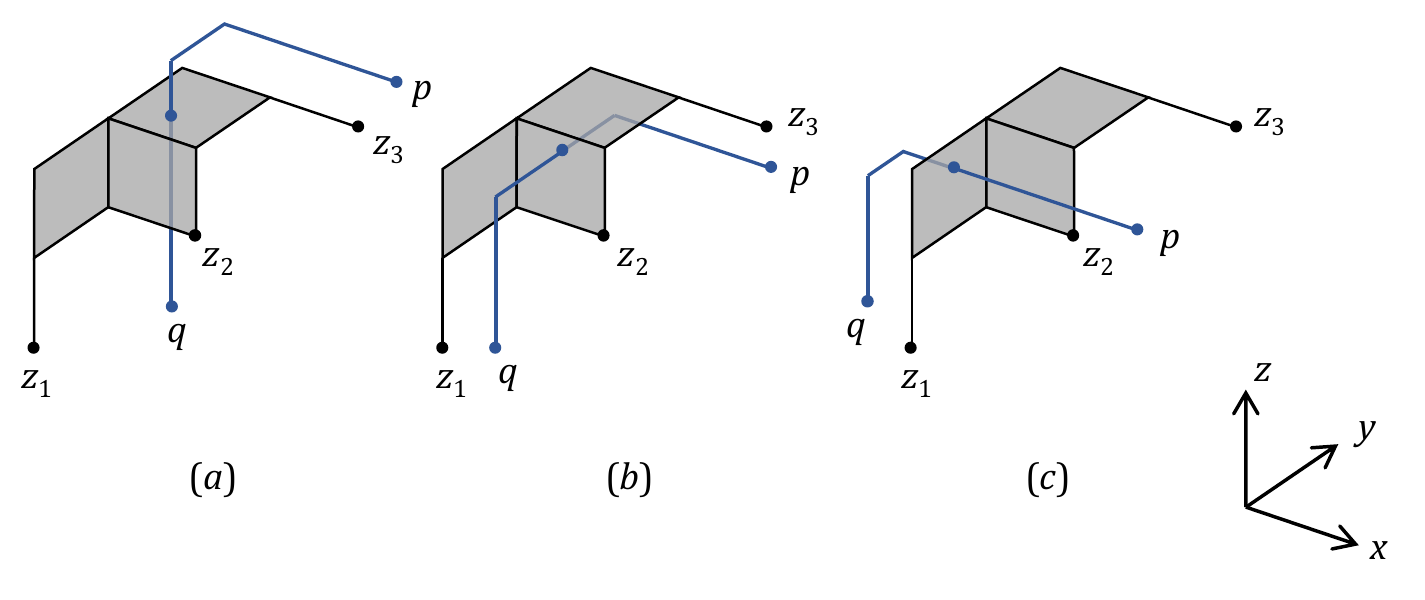}}
    \caption{A stair-path stabbing a stair-triangle ($3$-vertex stair-simplex) in dimension $3$.}\label{fig_r3}
\end{figure}

\subsection{Extending the stair-path to the boundary of the cube}
Without loss of generality, we can extend a given stair-path $qp$ until the two endpoints touch the boundary of the unit cube. This makes the calculations easier, since for each type $T$, there are two variables that can be set to $0$ or $1$.

Without loss of generality assume $d\notin T$, so $p_d\leq q_d$. Then $p_d$ can be extended to $0$. For $q$, the first coordinate that is elevated from $q$ to $p$, namely $\max \{i:q_i\leq p_i\}$, is the one that can be extended to $0$. If, on the other hand, $p_i \leq q_i$ for all $i$, then $q_1$ can be extended to $1$.

Hence, given the type $T$ of the stair-path $qp$, we proceed as follows: Say $d\notin T$ (otherwise, switch $p$ and $q$ and let $T$ be $\{ 1,\ldots,d\} \setminus T$). Then we let $p_d=0$. In addition, if $T \neq \emptyset$ then we let $q_{\max T}=0$, while if $T=\emptyset$ then we let  $q_1=1$. (As noted before, for the case of the stretched grid we can assume without loss of generality that $1\notin T$.)

\subsection{Maximum for the stretched grid}
Using the recursive function of Theorem \ref{theorem1} for the stretched grid, we get $2^{d-2}$ polynomial expressions in the coordinates $q_1, \ldots ,q_d,p_1, \ldots ,p_d$. We need to find the maximum for each expression.

\subsubsection{Results for dimension 3}

For $d=3$ there are two fundamentally different types: $T=\emptyset$ and $T=\{2\}$.

\paragraph{Type $T=\emptyset$.} For this case, we can let $q_1=1,p_3=0$, so in this case the function we want to maximize is
\begin{multline*}
F=\RecFSG_3(q,p)/3!= \\
(p_1-1) q_3^2\cdot\Big(p_1\big(p_2^2 (1-2 q_3)+q_2^2 (q_3-1)+p_2^3 q_3\big)+
   q_2 (q_2-1+q_3-q_2 q_3)\Big)
\end{multline*}
in the domain
\begin{equation*}
U=\{(p_1,p_2,q_2,q_3) \in \mathbb{R}^4:
0\leq p_1 \leq 1,
0\leq p_2 \leq q_2 \leq 1,
0\leq q_3 \leq 1 \}.
\end{equation*}

In order to find analytically the maximum of $F$ in $U$, we have to examine the interior of $U$ and its faces. 
$U$ is a 4-dimensional polytope, which according to the software ``polymake" \cite{Assarf2017} has 62 faces of various dimensions.

Alternatively, we could try using the function \verb|Maximize| of Mathematica 11, which finds the absolute maximum of a given function in a given range. Unfortunately, when given $F$ and $U$, \verb|Maximize| does not terminate in a reasonable amount of time. Hence, we employ a hybrid approach, dealing with the interior of $U$ by hand, and using \verb|Maximize| for the seven facets of $U$.

\paragraph{Maximum in the interior.}
If the maximum is attained by a point $(p_1, p_2, q_2, q_3)$ in the interior of $U$, then it must satisfy the following four equations:
\begin{align*} \label{eq: differential equations} 
    \frac{dF}{dq_2}&=(p_1-1)\overbrace{\big(1+2 (p_1-1) q_2\big)}^\text{I}(q_3-1)q_3^2=0,\\
    \frac{dF}{dq_3}&= (p_1-1)q_3 \\ 
    &\quad{}\cdot\overbrace{\Big(-(q_2-1) q_2(3 q_3-2)+p_1 \big(p_2^2 (2-6 q_3)+3 p_2^3 q_3+q_2^2(3q_3-2)\big)\Big)}^\text{II}=0,\\ 
    \frac{dF}{dp_1}&=q_3^2 \\
    &\quad{}\cdot\overbrace{\Big(q_2\big(1+2(p_1-1) q_2\big)(q_3-1)
    +(2 p_1-1) p_2^3 q_3 +p_2^2 \big(-1+p_1(2-4 q_3)+2q_3\big)\Big)}^\text{III}=0,\\
    \frac{dF}{dp_2}&=(p_1-1) p_1 p_2 q_3^2\overbrace{ \big(2+(3 p_2-4) q_3\big)}^\text{IV}=0.
\end{align*}
The solutions that satisfy $\{q_3=0\}$, or $\{p_1=1\}$, or both $\{q_3=1\}$ and $\{p_1=0\}$, or both $\{q_3=1\}$ and $\{p_2=0\}$, are irrelevant, since they give $F=0$. Therefore, any local maximum should satisfy the equations $II=0$ and $III=0$. One possibility is to satisfy $q_3=1$ and $IV=0$. The other three possibilities are to satisfy the equation $I=0$, as well as one of $IV=0$, $\{p_1=0\}$, $\{p_2=0\}$.
 These are the relevant solutions to the system:
\begin{gather*}
  \{p_1=(2q_2-1)/(2q_2), p_2=0, q_3=2/3\},\\
   \{p_1=0, p_2=0, q_2=1/2, q_3=2/3\},\\
  \{p_1=0, p_2=1/2, q_2=1/2, q_3=2/3\},\\
  \{p_1=1/2, p_2=1/2, q_2=1, q_3=4/5\},\\
  \{p_1=1/2, p_2=2/3, q_2=2/3,q_3=1\}.
\end{gather*}
After checking all these solutions, we get the maximum, $1/25$ by the solution $\{p_1=1/2, 
  p_2=1/2, q_2=1, q_3=4/5\}$.
 \paragraph{Maximum on the facets.}
  In order to find maximum on the facets of $U$, we used \verb|Maximize|. The results can be seen in Table \ref{table:1}.
 
\begin{table}
{\small 
  \begin{verbatim}
in:  F=(p1-1)q3^2*(p1(p2^2(1-2q3)+q2^2(q3-1)+p2^3q3)
      +q2(q2-1+q3-q2 q3))
  
in:  F1=F/.q2->1;
     Maximize[{F1,0<q3<1,0<p1<1,0<p2<1},{q3,p1,p2}]
out: {1/25,{q3->4/5,p1->1/2,p2->1/2}}

in:  F2=F/.q2->p2;
     Maximize[{F2,0<q3<1,0<p1<1,0<p2<1},{q3,p1,p2}]
out: {1/27,{q3->64/81,p1->10/37,p2->37/64}}

in:  F3=F/.p2->0;
     Maximize[{F3,0<q3<1,0<p1<1,0<q2<1},{q3,p1,q2}]
out: {1/27,{q3->2/3,p1->1/3,q2->3/4}}

in:  F4=F/.p1->1;
     Maximize[{F4,0<q3<1,0<p2<q2<1},{q3,q2,p2}]
out: {0,{q3->1/2,q2->3/4,p2->1/4}}

in:  F5=F/.p1->0;
     Maximize[{F5,0<q3<1,0<p2<q2<1},{q3,q2,p2}]
out: {1/27,{q3->2/3,q2->1/2,p2->1/4}}

in:  F6=F/.q3->1;
     Maximize[{F6,0<p1<1,0<p2<q2<1},{p1,q2,p2}]
out: {1/27,{p1->1/2,q2->27/32,p2->2/3}}

in:  F7=F/.q3->0;
     Maximize[{F7,0<p1<1,0<p2<q2<1},{p1,q2,p2}]
out: {0,{p1->1/2,q2->3/4,p2->1/4}} 
  \end{verbatim}}
\caption{Maximum on faces, $T=\emptyset$}
\label{table:1}
\end{table}

\paragraph{Type $T=\{2\}$.} Here one can proceed similarly. Here the maximum is also $1/25$, this time given by $qp=\{(2/3,0,4/5),(1/3,3/4,0)\}$. (Note that in this case, the maximum is in the interior of the domain, whereas in the case $T=\emptyset$ the maximum was on one of its facets.) The calculations can be found in the ancillary files of the arXiv version of this paper.

\subsubsection{Results for dimensions 4, 5, and 6}
For dimensions $d\ge 4$, the problem turns out to be too complex for the above approach. Therefore, we used the function \verb|NMaximize| of Mathematica, which searches for the absolute maximum numerically. The function \verb|NMaximize| provides four different numerical methods, called  \verb|NelderMead|, \verb|DifferentialEvolution|, \verb|SimulatedAnnealing|, and \verb|RandomSearch|. We tried all four of them. In dimensions $4$ and $5$ they all gave the same results, though not in dimension $6$.

In dimension $4$ there are four fundamentally different types of stair-paths. Table \ref{table:2} sums up the numerical results. The maximum among all the types is $0.00457936n^4$.

\begin{table}
\centering
\begin{tabular}{ccc}
    $T$ & $qp$ & maximum \\ \hline \\[-6pt]
    $\emptyset$ & $\{(1,0.99973,0.841676,0.824961),$ & 0.00456416\\ 
     & $(0.499854,0.57138,0.590885,0)\}$ & \\[6pt]
    $\{2\}$ & $\{(0.70017,0,0.841749,0.824908),$ & 0.00456416\\
      & $(0.400339,0.714113,0.590808,0)\}$ & \\[6pt]
    $\{3\}$ & $\{(0.666089,0.777112,0,0.827549),$ & \textbf{0.00457936} \\
      & $(0.395765,0.49188,0.794824,0)\}$ & \\[6pt]
    $\{2,3\}$& $\{(0.604237,0.491879,0,0.830208),$ &\textbf{0.00457936} \\
      & $(0.333913,0.77711,0.792279,0)\}$ & \\ 
  \end{tabular}
\caption{Results for dimension $4$, stretched grid.}
\label{table:2}
\end{table}

In dimension $5$ there are eight different types of stair-paths. The maximum among all of them is $0.000405335n^5$. See Table \ref{table:3}.

\begin{table}
\centering
  \begin{tabular}{ccc}
    $T$ & $qp$ & Maximum \\ \hline \\[-6pt]
    $\emptyset$ & $\{(1, 0.999998, 0.863413, 0.850444, 0.848693),$ & 0.000402464 \\
    & $(0.499999, 0.604701, 0.650764, 0.657374, 0)\}$ & \\[6pt]
    $\{2\}$ & $\{(0.71664,0,0.863421,0.850413,0.848695),$ & 0.000402464 \\
    & $(0.433377,0.697668,0.650744,0.657406,0)\}$ & \\[6pt]
    $\{3\}$ & $\{(0.675465,0.796913,0,0.850715,0.848819),$ & 0.00040419\\
     & $(0.428888,0.554061,0.78136,0.657527,0)\}$ & \\[6pt]
    $\{4\}$ & $\{(0.661946,0.786912,0.815446,0,0.850046),$ & 0.000404818 \\ 
     & $(0.425712,0.554894,0.590175,0.827888,0)\}$ & \\[6pt]
    $\{2,3,4\}$ & $\{(0.574368,0.554597,0.589951,0,0.853691),$ & 0.000404818 \\
     & $(0.337538,0.786894,0.815268,0.824237,0)\}$ & \\[6pt]
    $\{2,3\}$ & $\{(0.630424,0.478229,0,0.852643,0.849618),$ & 0.000404815 \\
     & $(0.387534,0.738946,0.779659,0.658423,0)\}$ & \\[6pt]
    $\{2,4\}$& $\{(0.622544,0.495106,0.817108,0,0.850792),$ & \textbf{0.000405335} \\ 
     & $(0.388108,0.740553,0.590829,0.826241,0)\}$ & \\[6pt]
    $\{3,4\}$ & $\{(0.612094,0.740651,0.590451,0,0.852007),$ & \textbf{0.000405335} \\ 
     & $(0.377364,0.494832,0.816924,0.824642,0)\}$ & \\[6pt]

  \end{tabular}
\caption{Results for dimension $5$, stretched grid.}
\label{table:3}
\end{table}

In dimension $6$ there are $16$ different types of stair-paths. Here, not all maximization methods gave the same result. The method \verb|DifferentialEvolution| gave the best results in all types. The maximum obtained is $0.0000291323n^6$, for types $T=\{2,3,5\}$ and $T=\{4,5\}$. For type $T=\{2,3,5\}$, the coordinates that give this maximum are
\begin{multline*}
   qp=\{(0.592993,0.545248,0.59284,0.843717,0,0.869422),\\
   (0.38511,0.750149,0.798446,0.658605,0.849763,0)\}. 
\end{multline*}

\section{Results for the stretched diagonal}\label{sec_diag}

In this section we prove the following:

\begin{theorem}\label{theorem2}
For every $d\ge 3$ there exists a stair-path $qp$ that stabs $n^d/(d+2)^{d-1}-o(n^d)$ stair-simplices spanned by the stretched diagonal $D_s(n)$.
\end{theorem}

Hence, for $d=4,5,6$ the stretched diagonal yields worse bounds for $c_{d,1}$ than the stretched grid. In this section we also prove that for $d=3$ the stretched diagonal yields $c_{3,1} \leq 1/25$, just like the stretched grid.

\subsection{A recursive formula for a special case}

Let $\mathcal{D}$ be the uniform distribution in $[0,1]^d$. Every element $A=(a_1, \ldots, a_d) \in [0,1]^d$ represents a $d$-tuple of points $\vec{a}_1,\ldots,\vec{a}_d$ where $\vec a_i = (a_i, \ldots, a_i)\in[0,1]^d$ for each $i$. These points span the stair-simplex $S(A)=\stconv\{\vec{a}_1,\ldots,\vec{a}_d\}$.

Given two points $q,p\in [0,1]^d$ with $q=(q_1, \ldots, q_d)$, $p=(p_1, \ldots, p_d)$, let $\FSD(q,p)$ be the measure of all the $d$-tuples $A=(a_1, \ldots, a_d)\in \mathcal{D}$ that satisfy the following two conditions:
\begin{enumerate}
    \item $a_1<\cdots<a_d$,
    \item $\stconv\{q,p\}\cap\stconv S(A)\neq\emptyset$.
\end{enumerate}

The connection between $\FSD$ and the stretched diagonal is as follows: Let $n$ be large enough, and let $D_s(n)$ be the $n$-point stretched diagonal, and let $q', p'$ be defined from $q,p$ as before. Then the probability that a random stair-simplex spanned by $D_s(n)$ intersects the stair-path $q'p'$ is very close to $d!\FSD(q,p)$, and hence, by the Transference Lemma, the number of simplices spanned by $D_s(n)$ that intersect the segment $q'p'$ is $\FSD(q,p)n^d$ plus lower-order terms.

We prove Theorem~\ref{theorem2} by calculating $\FSD$ for a certain sub-type of stair-path that belongs to the type $T=\emptyset$. Specifically, we will calculate $\FSD(q,p)$ for the special case where the points $p=(p_1, \ldots, p_d)$ and $q=(q_1, \ldots, q_d)$ satisfy the following conditions:
\begin{equation}\label{eq_cond_FSD}
    p_1\leq p_2\leq \cdots\leq p_d,\qquad q_2\leq q_3\leq \cdots \leq q_d\leq q_1=1,\qquad p_i\leq q_i \text{ for all $i$}.
\end{equation}

Let $p,q\in[0,1]^d$ satisfy conditions (\ref{eq_cond_FSD}), and let $Y=\{q,p\}$. We define $\RecFSD_d(q,p)$ as the measure of $(d-1)$-tuples $a_1, \ldots, a_{d-1}\in [0,1]$ satisfying the following two conditions:
\begin{enumerate}
\item $a_1< a_2 < \cdots < a_{d-1}<q_d$,
\item $\stconv(\overline Y) \cap \stconv(\overline Z)\neq\emptyset$, where $Z=\{\vec{a}_1, \ldots, \vec{a}_{d-1}\}$.
\end{enumerate}
(For $p,q$ not satisfying conditions (\ref{eq_cond_FSD}), $\RecFSD_d(q,p)$ is undefined.) Note that if $p,q$ satisfy (\ref{eq_cond_FSD}), then so do $\overline p,\overline q$.

\begin{observation}\label{obs_FSD}
Let $p,q\in[0,1]^d$ satisfy conditions (\ref{eq_cond_FSD}). Then,
\begin{equation*}
\FSD_d(q,p)=\RecFSD_{d+1}(q\times 1,p\times 1).
\end{equation*}
\end{observation}

\begin{lemma}\label{Lemma_RecFSD}
For $p,q\in [0,1]^d$ satisfying conditions (\ref{eq_cond_FSD}), $\RecFSD_d(q,p)$ is given recursively as follows:
\begin{align*}
\RecFSD_2(q,p)&=q_2-p_1,\\
\RecFSD_{d+1}(q,p)&=(q_d-p_d)p_1\prod_{i=2}^{d-1}{(p_i-p_{i-1})}\\
&\qquad{}+(q_{d+1}-q_d)\RecFSD_d(\overline{q},\overline{p}), \qquad \text{for } d\geq 2.
\end{align*}
\end{lemma}
For example, putting together Observation \ref{obs_FSD} and Lemma \ref{Lemma_RecFSD}, we get that for $p,q$ satisfying (\ref{eq_cond_FSD}) we have
\begin{align*}
 \FSD_1(q,p)&=1-p_1,\\
 \FSD_2(q,p)&=(q_2-p_2)p_1+(1-q_2)(q_2-p_1),\\
 \FSD_3(q,p)&=(q_3-p_3)p_1(p_2-p_1)+(1-q_3)\bigl((q_2-p_2)p_1+(q_3-q_2)(q_2-p_1)\bigr).
\end{align*}
\begin{proof}[Proof of Lemma \ref{Lemma_RecFSD}]
By induction on $d$.

\textbf{Base case:}
When $d=2$, we need $a_1 <q_2$ and $p_1 \leq a_1\leq q_1=1$. Therefore, we need $p_1 \leq a_1\leq q_2$. The measure of numbers $a_1\in [0,1]$ satisfying this condition is $q_2-p_1$.

\textbf{Recursive case:} Suppose we are in dimension $d+1$, and let $Y=\{q,p\}$ with $q,p\in[0,1]^{d+1}$ and $Z=\{\vec{a}_1, \ldots, \vec{a}_d\}$ where $\vec {a}_i=(a_i,\ldots,a_i)\in [0,1]^{d+1}$. We need the numbers $a_i$ to satisfy the following two separate conditions:
\begin{enumerate}
\item $a_d<q_{d+1}$.
\item $\stconv(\overline Y) \cap \stconv(\overline Z)\neq\emptyset$ with $a_1< a_2 < \cdots < a_{d}$.
\end{enumerate}

Let us calculate the measure of tuples satisfying the second condition.
As in the recursive formula for the stretched grid, here there are two possibilities, according to whether the highest point in dimension $d$ is $\overline q$ or $\overline{\vec{a}_d}$. In the first case, we must have $a_d<q_d$ (which automatically implies  $a_d<q_{d+1}$). In addition, by Lemma \ref{Lemma 5.4}, we need to have $p_d<a_d$, and after ``discarding" point $q$ and projecting down to dimension $d-1$, we need to have $p'=\overline{\overline{p}} \in \stconv{\bigl(\overline{\overline{ Z}}\bigr)}$. For this, we apply Lemma \ref{Lemma x in stconv}. The set $\overline{\overline{Z}}$ must contain a point of type $j$ with respect to $p'$ for every $j=0,\ldots,d-1$. We also need the coordinates $a_i$ to be in increasing order. Therefore, for type $d-1$, the $(d-1)$-st coordinate of $\vec{a}_d$ should be higher than $p_{d-1}$, that is $p_{d-1}<a_d$ but since we demand $p_d<a_d$ this condition is irrelevant. For type $d-2$, the $(d-2)$-nd coordinate of $\vec{a}_{d-1}$ should be higher than $p_{d-2}$, that is $p_{d-2}<a_{d-1}$. In addition, $\vec{a}_{d-1}$ should be lower than $p$ in the higher coordinates, and therefore $p_{d-2}<a_{d-1}<p_{d-1}$. And so on. In general, for every type $j=1,\ldots, d-2$ we must have $p_j < a_{j+1} < p_{j+1}$.
For type $0$, the lowest point $\vec{a}_1$ must satisfy $a_1 < p_1$.
To sum up, the measure in the first case is
\begin{equation*}
(q_d-p_d)p_1\prod_{i=2}^{d-1}{(p_i-p_{i-1})}.
\end{equation*}

In the second case we must have $a_d>q_d$. Together with the condition $a_d<q_{d+1}$, this implies $q_d<a_d<q_{d+1}$.  In addition, by Lemma \ref{Lemma 5.4}, we need to have $a_{d-1}<q_d$, and after ``discarding" point $\vec a_d$ and projecting down to dimension $d-1$, we need to have an intersection between $\stconv\bigl\{\overline{\overline{p}},\overline{\overline{q}}\bigr\}$ and $\stconv\bigl\{\overline{\overline{\vec a_1}}, \ldots, \overline{\overline{\vec a_{d-1}}}\bigr\}$. But these are exactly conditions 1 and 2 above, one dimension lower. Therefore, the measure in the second case is
\begin{equation*}
(q_{d+1}-q_d)\RecFSD_{d-1}(\overline{q},\overline{p}).
\end{equation*}
\end{proof}

\begin{proof}[Proof of Theorem \ref{theorem2}]
Let us take the following stair-path:
\begin{equation*}
qp=\left\{\left(1,\frac{3}{d+2},\frac{4}{d+2},\ldots,\frac{d+1}{d+2}\right),\left(\frac{1}{d+2},\frac{2}{d+2},\ldots ,\frac{d-1}{d+2},\frac{d-1}{d+2}\right)\right\}.
\end{equation*}
The points $q, p$ satisfy conditions (\ref{eq_cond_FSD}). Hence,
\begin{equation*}
\FSD_d(q,p)=\RecFSD_{d+1}(q\times 1,p\times 1)=\frac{2}{(d+2)^d}+\frac{1}{d+2}\cdot \RecFSD_d(q,p)
\end{equation*}
It follows by induction on $d$ that, if $q',p'\in\mathbb{R}^d$ have the form $q'=(1, 3/c, 4/c,\allowbreak \ldots, (d+1)/c)$ and $p'=(1/c, 2/c, \ldots, (d-1)/c, k)$ for some $c,k$, then $\RecFSD_d(q',p') = d/c^{d-1}$. Therefore, in our case, $\RecFSD_d(q,p) = 1/(d+2)^{d-1}$. The claim follows.
\end{proof}

\subsection{Dimension 3}
In order to find the maximum for the stretched diagonal in $d=3$, we examine all different possible types of stair-paths, each one having its own expression $F$ and domain $U$. See Table \ref{table:4}. The expressions $F$ can be derived from Lemma \ref{Lemma 5.4}, as in previous sections, or they can be derived more directly as follows: Let $\vec a = (a,a,a)$, $\vec b = (b,b,b)$, $\vec c = (c,c,c)$, with $0<a<b<c<1$, and consider the stair-simplex $S=\stconv(\vec a, \vec b, \vec c)$. $S$ contains three axis-parallel rectangles with different orientations: rectangle $R_1$ with opposite corners $(a,a,b)$, $(a,b,c)$, rectangle $R_2$ with opposite corners $(a,b,b)$, $(b,b,c)$, and rectangle $R_3$ with opposite corners $(a,b,c)$, $(b,c,c)$. Furthermore, the stair-path $qp$ is composed of three axis-parallel segments with three different orientations. In order for the stair-path $qp$ to intersect $S$, one of the former's segments must intersect one of the latter's rectangles. Hence, the numbers $a,b,c$ must satisfy some inequalities depending on the coordinates $q, p$, which are not hard to work out.

We calculated the maximum in each case using \verb|Maximize|. In contrast to the stretched grid, where the degree of $F$ was $8$, here the degree of $F$ is only $3$, so \verb|Maximize| had no problem finding the maximum quickly.  The maximum is $(1/25)n^3$, see Table \ref{table:5}.
\begin{table}
   $T=\emptyset$, $0<p_1 < p_2 < q_2 < q_3<1$, \\
    $F_1=p_1(q_2-p_2)(1-q_3)+
 p_1 (p_2-p_1) (q_3-p_2)+(q_2-p_1) (1-q_3) (q_3-q_2)$ \\[6pt]
    $T=\emptyset$,  $0<p_1 < p_2 <q_3< q_2<1$, \\
    $F_2=p_1 (p_2-p_1) (q_3-p_2) + p_1 (1 - q_3) (q_3-p_2)$\\[6pt]
    $T=\emptyset$, $0<p_2 < p_1<q_3 < q_2<1$, \\ 
    $F_3=p_1 (1 - q_3) (q_3-p_1)$\\[6pt]
  $T=\emptyset$, $0<p_2 < p_1<q_2 < q_3<1$,   \\ 
  $F_4=p_1 (q_2-p_1) (1 - q_3) + (q_2-p_1) (1 - q_3) (q_3-q_2)$\\[6pt]
  $T=\{1\}$, $0<p_1 < p_2 < q_2 < q_3<1$,\\
    $F_5=p_1 (q_2-p_2) (1 - q_3) + p_1 (p_2-p_1) (q_3-p_2) + 
 p_1 (1 - q_3) (q_3-q_2)$\\[6pt]
 $T=\{1\}$, $0<p_1 < p_2 <q_3 < q_2<1$,\\
 $F_6=p_1 (p_2-p_1) ( q_3-p_2) + p_1 (1 - q_3) ( q_3-p_2)$\\[6pt]
 $T=\{1\}$, $0<p_2 < p_1<q_3 < q_2<1$,\\
 $F_7=p_1 (1 - q_3) (q_3-p_1)$\\[6pt]
 $T=\{1\}$, $0<p_2 < p_1 < q_2 < q_3<1$,\\
 $F_8=p_1 (q_2-p_1) (1 - q_3) + p_1 (1 - q_3) (q_3-q_2)$\\[6pt]
 $T=\{1\}$, $0<p_2 < q_2 < p_1<1$, $0<q_2 < q_3<1$,\\
 $F_9=q_2 (1 - q_3) (q_3-q_2)$\\[6pt]
 $T=\{2\}$, $0<p_1 < q_1 < p_2 < q_3<1$,\\
 $F_{10}=(p_2 - q_1) q_1 (1 - q_3) + 
 p_1 (p_2-p_1) (q_3-p_2) + (q_1-p_1) (1 - q_3) (q_3-p_2)$\\[6pt]
 $T=\{2\}$, $0<p_1<p_2 < q_1<1$, $0<p_2 < q_3<1$,\\
 $F_{11}=p_1 (p_2-p_1) (q_3-p_2) + (p_2-p_1) (1 - q_3) (q_3-p_2)$\\[6pt]
 $T=\{2\}$, $0<p_1 < q_1<q_3 < p_2<1$,\\
 $F_{12}=q_1 (1 - q_3) (q_3-q_1)$\\[6pt]
 $T=\{1,2\}$, $0<q_1 < p_1 < p_2 < q_3<1$,\\
 $F_{13}=(p_2 - q_1) q_1 (1 - q_3) + p_1 (p_2-p_1) (q_3-p_2) + (p_1 - q_1) (1 - q_3) (q_3-p_2)$\\[6pt]
 $T=\{1,2\}$, $0<q_1 < p_2<p_1<1$, $0<p_2 < q_3<1$,\\
 $F_{14}=(p_2 - q_1) q_1 (1 - q_3) + (p_2 - q_1) (1 - q_3) (q_3-p_2 )$\\[6pt]
 $T=\{1,2\}$, $0<q_1 < p_1<1$, $0<q_1<q_3 < p_2<1$,\\
 $F_{15}=q_1 (1 - q_3) (q_3-q_1)$
\caption{Different functions for the stretched diagonal in dimension $3$.}
\label{table:4}
\end{table}

\begin{table}
\centering
\begin{tabular}{ccc}
    $F$ & $qp$ & maximum \\ \hline\\[-6pt]
$F_1$&$(1,3/5,4/5),(1/5,2/5,0)$ & $\mathbf{1/25}$\\ 
$F_2$&$(1,59/64,27/32),(1/3,49/96,0)$  & $1/27$ \\ 
    $F_3$& $(1,27/32,2/3),(1/3,5/32,0)$ & $1/27$ \\ 
    $F_4$& $(1,1/2,2/3),(1/6,5/64,0)$ &  $1/27$ \\ 
    $F_5$& $(0,11/16,27/32),(1/3,49/96,0)$ & $1/27$ \\ 
    $F_6$&$(0,59/64,27/32),(1/3,49/96,0)$  & $1/27$ \\ 
    $F_7$& $(0,27/32,2/3),(1/3,5/32,0)$ & $1/27$ \\ 
    $F_8$&$(0,1/2,2/3),(1/3,5/32,0)$  & $1/27$ \\ 
    $F_9$& $(0,1/3,2/3),(11/16,5/32,0)$ & $1/27$ \\ 
    $F_{10}$& $(2/5,0,4/5),(1/5,3/5,0)$ & $\mathbf{1/25}$ \\ 
    $F_{11}$& $(11/16,0,3/4),(1/12,5/12,0)$ &  $1/27$ \\ 
    $F_{12}$& $(1/3,0,2/3),(5/32,27/32,0)$ &  $1/27$ \\ 
    $F_{13}$& $(1/5,0,4/5),(2/5,3/5,0)$ & $\mathbf{1/25}$ \\ 
    $F_{14}$& $(5/32,0,2/3),(3/4,47/96,0)$ & $1/27$ \\ 
    $F_{15}$& $(1/3,0,2/3),(11/16,27/32,0)$ & $1/27$ \\ 
  \end{tabular}
\caption{Maximum of each function in Table \ref{table:4}.}
\label{table:5}
\end{table}

\section{Discussion and future work}\label{sec_final}

Since in dimension $d=3$, the stretched grid and the stretched diagonal yield the same upper bound of $n^3/25$ (which is known to be tight), we were expecting the same to happen in higher dimensions. We were surprised to find this not to be the case. Also surprising is the fact that the bounds obtained for $d\ge 4$ do not seem to be rational. Running \verb|NMaximize| with higher precision, we find the bound for $d=4$ to be $0.004579364805943860006\ldots$.

In order to gain more confidence in our numerical results, we re-ran the stretched-grid numerical maximization experiments in a newer version of Mathematica (12.3) as well as in Python using SciPy's \verb|differential_evolution| function. We got the same results. See the ancillary files of the arXiv version.

The main open problem is to find the exact value of the constants $c_{d,1}$. Since in dimension $4$, the stretched grid and the stretched diagonal do not give the same value, we are not sure that the value given by the stretched grid is tight.

Another interesting problem is to study the corresponding variant of the Second Selection Lemma, in which we look for a line that stabs many simplices from a given subset of $X$-simplices. One could also study variants in which a line stabs geometric objects other than simplices.

\paragraph{Acknowledgements.} Thanks to Elad Horev, Rom Pinchasi, and the anonymous referee for their useful comments.

\bibliographystyle{plainurl}
\bibliography{ref}

\end{document}